\documentclass[12pt]{amsart}
\usepackage{amssymb,latexsym}
\usepackage{amsfonts}
\usepackage{amsmath}
\usepackage[colorlinks,linkcolor=blue,anchorcolor=blue,citecolor=blue]{hyperref}
\usepackage{algorithm}
\usepackage{enumerate}
\usepackage{algpseudocode}
\usepackage{verbatim}
\usepackage{graphicx}

\newcommand{\cS}{{\mathcal{S}}}
\newcommand{\tr}{{\rm Tr}}

\newcommand{\gZ}{{\mathrm Z}}

\newcommand{\F}{{\mathbb F}}

\newcommand{\gC}{\mathrm{C}}
\newcommand{\cM}{\mathcal{M}}
\newcommand{\cC}{{\mathcal C}}

\newtheorem{theorem}{Theorem}[section]
\newtheorem{lemma}[theorem]{Lemma}

\theoremstyle{definition}
\newtheorem{definition}[theorem]{Definition}

\newtheorem{corollary}[theorem]{Corollary}

\theoremstyle{remark}

\numberwithin{equation}{section}

\begin{document}

\title{An Unified Approach on Constructing of MDS Self-dual Codes via Reed-Solomon Codes}

\author{Aixian Zhang}
\address{Department of Mathematical Sciences, Xi'an  University of Technology,
Shanxi, 710054, China.}
\email{zhangaixian1008@126.com}

\author{Keqin Feng}
\address{Department of Mathematical
Sciences, Tsinghua University,
 Beijing, 100084, China. }
\email{fengkq@tsinghua.edu.cn}

%    General info

%\date{Januar1 1, 2001 and, in revised form, June 22, 2001.}



\begin{abstract}
Based on the fundamental results on MDS self-dual codes over finite fields constructed via generalized
Reed-Solomon codes \cite{JX} and extended generalized Reed-Solomon codes \cite{Yan}, many series of MDS self-dual codes with different length have been obtained recently by a variety of constructions and individual computations.
In this paper, we present an unified approach to get several previous results with concise statements and
simplified proofs, and some new constructions on MDS self-dual codes. In the conclusion section we raise two
open problems.
\end{abstract}

\maketitle

%% The correct journal style for \specialsection is all uppercase; a known bug
%% in amsart.cls prevents this, so input must be uppercase until it is fixed.
%\specialsection*{This is a Special Section Head}

%%%%%%%%%%%%%%%%%%%%%%%%%%%%%%%%%%%%%%%%%%%%%%%%%%%%%%%%%%%%%%%%%%%%%%%%
%\footnote{Here is an example of a footnote. Notice that this footnote
%text is running on so that it can stand as an example of how a footnote
%with separate paragraphs should be written.
%\par
%And here is the beginning of the second paragraph.}%
%%%%%%%%%%%%%%%%%%%%%%%%%%%%%%%%%%%%%%%%%%%%%%%%%%%%%%%%%%%%%%%%%%%%%%%%

\section{Introduction}\label{sec-one}
Let $\cC$ be a linear code with parameters $[n,k,d]_q.$ Namely, $\cC$
is a $\F_q$-subspace of $\F^n_q$ with dimension $k$ and the minimum (Hamming) distance $d.$
We have the Singleton bound $d \leq n-k+1.$ $\cC$ is called MDS code if $d=n-k+1.$
The dual code of $\cC$ is defined by
$$
\cC^{\perp}=\{ v \in \F^n_q: (v,c)=0 \ \mbox{for all} \ c \in \cC \}
$$
where for $v=(v_1, v_2,\ldots,v_n)$ and $c=(c_1, c_2,\ldots,c_n),(v,c)=\sum^n\limits_{i=1}v_i c_i \in \F_q $
is the usual inner product in $\F^n_q$. $\cC^{\perp}$ is a linear code with parameters
$[n,k^{\perp}]_q$ where $k^{\perp}=n-k.$ If $\cC=\cC^{\perp},\cC$ is called self-dual.
For self-dual code $\cC$, the length $n$ of the code $\cC$ should be even by $k=n-k.$

MDS self-dual codes have attracted a lot of attention in recent years by their theoretical interests
in coding theory, many applications in cryptography and combinatorics. For such codes, $k=\frac{n}{2}$
and $d=\frac{n}{2}+1$ are determined by the length $n$. One of the basic problems is that for a fixed
prime power $q,$ which even number $n \geq 2$ can be the length of an MDS self-dual code over finite field $\F_q$ ?

This problem has been solved in the case $q=2^m$ by Grassl and Gulliver \cite{GG}. From now on, we consider $q=p^m$
where $p$ is an odd prime number and $m \geq 1$. Several families of MDS self-dual codes over $\F_q$ have been
constructed with length $n$ satisfying certain conditions by using generalized Reed-Solomon (GRS for short) codes
and extended generalized Reed-Solomon (EGRS for short) codes \cite{FF}-\cite{GK},\cite{GKL},\cite{JX}-\cite{LLL},\cite{Yan},
orthogonal designs \cite{HK,SSSC}, extended cyclic duadic codes and negacyclic codes \cite{Guenda}. Roughly speaking,
the first approach is to look for the GRS codes and EGRS codes as candidates of MDS codes, then to find sufficient
conditions satisfied by length $n$ such that the codes are self-dual. The last two approaches are to look for
the self-dual codes given by orthogonal designs and (nega-)cyclic codes and select ones being MDS codes.
A table of MDS self-dual codes over $\F_q$ is provided in \cite{SSSC} for length $n \leq 12$ and odd prime number $p \leq 109.$

In this paper, we focus on the first approach to consider the MDS codes given by GRS and EGRS codes. L.Jin,  C.Xing \cite{JX}
and H.Yan \cite{Yan} present basic results on necessary and sufficient conditions for the GRS codes and EGRS codes
being self-dual respectively. By using these conditions, several families of MDS self-dual codes with various length $n$
have been constructed \cite{FF,FLLL,JX,LLL,Yan} with careful computations.

In Section \ref{sec-two}, we introduce the basic results given in \cite{JX} and \cite{Yan}, but we provide an unified point of view on
the constructions given there. In Section \ref{sec-three}, we select some known results stated and proved with our approach.
In Section \ref{sec-four}, we construct new MDS self-dual codes with our approach. In Section \ref{sec-five}, we make conclusion and raise two open problems.

\section{Constructions of MDS Self-dual codes via GRS Codes}\label{sec-two}

\subsection{ (Extended) Generalized RS Codes}

\begin{definition}
Let $ q=p^m, \cS=\{a_1, a_2, \ldots, a_n\}$ be a subset of $\F_{q}$ with $n$
distinct elements (so that $n \leq q$), $v_1, v_2, \ldots, v_n$ be nonzero elements in $\F_q$
(not necessarily distinct), $v=(v_1, v_2, \ldots, v_n ).$ For $ 1 \leq k \leq n-1,$ the GRS code
is defined by
$$
\cC_{grs}(\cS,v,q)=\left\{ \begin{array}{ll} c_f=(v_1f(a_1),v_2f(a_2),\ldots,v_nf(a_n)) \in \F^n_q: \\
f(x) \in \F_q[x], \deg f \leq k-1 \end{array}\right\}.
$$
This is an MDS (linear) code over $\F_q$ with parameters $[n,k,d]_q, d=n-k+1.$ The extended GRS code is defined by
$$
\cC_{egrs}(\cS,v,q)=\left \{\begin{array}{ll} c_f=(v_1f(a_1),v_2f(a_2),\ldots,v_nf(a_n),f_{k-1}) \in \F^{n+1}_q: \\ f(x) \in \F_q[x], \deg f \leq k-1 \end{array} \right \}
$$
where $f_{k-1}$ is the coefficient of $x^{k-1}$ in $f(x).$ This is also an MDS code over $\F_q$ with parameters
$[n+1,k,d]_q, d=n-k+2.$
\end{definition}

A sufficient condition on set $\cS$ has been given in \cite{JX} and \cite{Yan} for $\mathcal{C}_{grs}$
and $\mathcal{C}_{egrs}$    being
self-dual. From the proofs we can see that the sufficient condition is also necessary.

For $\cS=\{a_1, a_2, \ldots, a_n\} \subseteq \F_q,$  we denote
$$
\Delta_{\cS}(a_i)=\prod_{\scriptstyle 1 \leq j \leq n
\atop \scriptstyle j \neq i}(a_i -a_j) \in \F^{\ast}_q.
$$
Let $\eta_q: \F^{\ast}_q \rightarrow \{\pm 1\}$ be the quadratic (multiplicative) character of $\F_q.$
Namely, for $b \in \F^{\ast}_q$,
$$
\eta_q(b) =  \left \{
\begin{array}{ll}
1, & \mbox{if} \ b \ \mbox{is a square in } \F^{\ast}_q, \\
-1, & \mbox{otherwise} .
\end{array}
\right.
$$

Now we introduce the basic results given in \cite{JX} and \cite{Yan}.

\begin{theorem}\label{thm-main}
Let $a_1, a_2, \ldots, a_n$ be distinct elements in $\F_q, \cS=\{a_1, a_2, \ldots, a_n\}.$

(1) \ (\cite{JX}) Suppose that $n$ is even. There exists $v=(v_1, v_2, \ldots, v_n ) \in (\F^{\ast}_q)^{n}$ such that the (MDS)
code $\cC_{grs}(\cS,v,q)$ is self-dual if and only if all $\eta_q(\Delta_{\cS}(a)) \ (a \in \cS)$
are the same.

(2)\ (\cite{Yan}) Suppose that $n$ is odd. There exists $v=(v_1, v_2, \ldots, v_n )\in (\F^{\ast}_q)^{n}$ such that the (MDS) code
$\cC_{egrs}(\cS,v,q)$ is self-dual code with length $n+1$ if and only if
$\eta_q (-\Delta_{\cS}(a))=1$ for all $a \in \cS.$
\end{theorem}

\begin{definition}
Let $\Sigma(q)$ be the set of all even number $n \geq 2$ such that there exists MDS self-dual code over $\F_q$
with length $n$. Let $\Sigma(g,q)$ and $\Sigma(eg,q)$ be the set of all even number $n\geq 2$ such that there
exists MDS self-dual code over $\F_q$ with length $n$ constructed by generalized RS code (Theorem \ref{thm-main} (1)) and extended generalized RS code (Theorem \ref{thm-main} (2)) respectively. Namely,
$$
\Sigma(g,q)=\left \{ n:      \begin{array}{ll}
2 \mid n \geq 2, \mbox{there exists a subset } \ \cS \ \mbox{of} \ \F_q, |\cS|=n,  \\
\mbox{such that all} \ \eta_q(\Delta_{\cS}(a)) \  (a \in \cS)\ \mbox{are the same}.
\end{array} \right\}.
$$

$$
\Sigma(eg,q)=\left\{ n:      \begin{array}{ll}
2 \mid n \geq 2, \mbox{there exists a subset } \ \cS \ \mbox{of} \ \F_q, |\cS|=n-1,  \\
\mbox{such that all} \ \eta_q(-\Delta_{\cS}(a))=1 \ \mbox{for all } \ a \in \cS.
\end{array}\right\}.
$$
\end{definition}
We have $\Sigma(g,q)  \cup \Sigma(eg,q) \subseteq \Sigma(q).$ As a direct consequence of Theorem \ref{thm-main}, we have
\begin{corollary} (\cite{JX})\label{thm-coro}
Suppose that $r=p^m$ and $q=r^2,$ where $p$ is an odd prime number and $m \geq 1.$ Then for any even number $n,
 n \in \Sigma(g,q)$ if $ 2 \leq n \leq r-1$ and $n \in \Sigma(eg,q)$ if $ 4 \leq n \leq r+1.$
 Therefore for all even number $n, 2 \leq n \leq r+1, n \in \Sigma(q).$
\end{corollary}

\begin{proof}
Suppose that $2 \mid n \leq r-1.$ Let $\cS$ be any subset of $\F_r$ with size $n.$
Then for each $a \in \cS,\Delta_{\cS}(a)=\prod\limits_{\scriptstyle b \in \cS
\atop \scriptstyle b \neq a}(a-b) \in \F^{\ast}_r.$ Since $\F_q$ is the quadratic extension of $\F_r,$
each non-zero element of $\F_r$ is a square in $\F^{\ast}_q.$ Therefore $\eta_q(\Delta_{\cS}(a))=1$
for all $a \in \cS.$ By Theorem \ref{thm-main} (1), we get $n \in \Sigma(g,q).$ Similarly we get that $n \in \Sigma(eg,q)$
for even number $n, 4  \leq n \leq r+1$ by Theorem \ref{thm-main} (2) and taking a subset $\cS$ of $\F_r$
with size $n-1.$
\end{proof}

\section{Binomials and Linearized Polynomials}\label{sec-three}
In general case, we should choose a subset $\cS$ of $\F_q$ and compute $\Delta_{\cS}(a) \ (a \in \cS)$
carefully to verify the conditions stated in Theorem \ref{thm-main}. The following simple fact can be used to
simplify the computation of $\Delta_{\cS}(a)$ in many cases. For a polynomial
$f(x)=\sum a_i x^i \in \F_q[x],$ we denote the derivative of $f(x)$ by $f^{\prime}(x)=\sum a_i i x^{i-1} \in \F_q[x].$

\begin{lemma}\label{thm-three}
(1) Let  $\cS=\{a_1, a_2, \ldots, a_n\}$ be a subset of $\F_q, f_{\cS}(x)=\prod\limits_{a \in \cS}(x-a).$
Then for any $ a \in \mathcal{S},\Delta_{\mathcal{S}}(a)=f^{\prime}_{\cS}(a).$

(2) Let $\cS_1$ and  $\cS_2$ be disjoint subsets of $\F_q, \cS=\cS_1 \bigcup \cS_2,$
$f_{\cS_i}(x)=\prod\limits_{a \in \cS_i}(x-a) \ (i=1,2).$ Then for $b \in \cS,$

$$
\Delta_{\cS}(b) =  \left \{
\begin{array}{ll}
\Delta_{\cS_1}(b)f_{\cS_2}(b), & \mbox{if} \ b \in \cS_1, \\
\Delta_{\cS_2}(b)f_{\cS_1}(b), & \mbox{if} \ b \in \cS_2.
\end{array}
\right.
$$

(3) Let $\mathcal{M} \subseteq \F_q, g(x)$ be a monic polynomial of degree $d$ in $\F_q[x].$
Assume that for each $a \in \mathcal{M}, g(x)=a$ has exactly $d$ distinct solutions $x$ in $\F_q.$
Namely, the size of $g^{-1}(a)= \{ b \in \F_q : g(b)=a\}$ is $d$ for all $a \in \mathcal{M}.$
Let $\cS=\bigcup\limits_{a \in \mathcal{M}}g^{-1}(a).$ Then for $b \in \mathcal{S},$
we have $\Delta_{\cS}(b)=\Delta_{\mathcal{M}}(a)g^{\prime}(b)$ where $a=g(b) \in \mathcal{M}.$
\end{lemma}

\begin{proof}
(1) From $ f^{\prime}_{\cS}(x)=(\prod^n\limits_{\lambda=1}(x-a_{\lambda}))^{\prime}
=\sum^n\limits_{l=1}(\prod^n\limits_{\scriptstyle \lambda=1
\atop \scriptstyle \lambda \neq l}(x-a_{\lambda})),$ we get, for $ 1 \leq i \leq n,$
$$
f^{\prime}_{\cS}(a_i)=\sum^n_{l=1}(\prod^n\limits_{\scriptstyle \lambda=1
\atop \scriptstyle \lambda \neq l}(a_i -a_{\lambda}))=\prod^n\limits_{\scriptstyle \lambda=1
\atop \scriptstyle \lambda \neq i}(a_i -a_{\lambda})=\Delta_{\cS}(a_i).
$$

(2) From $f_{\cS}(x)=\prod\limits_{a \in \cS}(x-a)=f_{\cS_1}(x)f_{\cS_2}(x),$
we get $$f^{\prime}_{\cS}(x)=f^{\prime}_{\cS_1}(x)f_{\cS_2}(x)+f_{\cS_1}(x)f^{\prime}_{\cS_2}(x).$$
For $b \in \cS_1, f_{\cS_1}(b)=0$
and then
$$\Delta_{\cS}(b)=f^{\prime}_{\cS}(b)=f^{\prime}_{\cS_1}(b)f_{\cS_2}(b)=\Delta_{\cS_1}(b)f_{\cS_2}(b).$$
Similarly, for $b \in \cS_2, \Delta_{\cS}(b)=\Delta_{\cS_2}(b)f_{\cS_1}(b).$

(3) For $a \in \mathcal{M}, g^{-1}(a)$ is the set of zeros of $g(x)=a.$ Namely, $g(x)-a=\prod\limits_{b \in g^{-1}(a)}(x-b).$
We get
$$
f_{\cS}(x)=\prod_{a \in \mathcal{M}}\prod_{b \in g^{-1}(a)}(x-b)=\prod_{a \in \mathcal{M}}(g(x)-a)=f_{\mathcal{M}}(g(x)),
$$
where $f_{\mathcal{M}}(x)=\prod\limits_{a \in \mathcal{M}}(x-a).$ Then from $f^{\prime}_{\cS}(x)=f^{\prime}_{\mathcal{M}}(g(x))g^{\prime}(x),$
we get, for $ b \in \cS, a=g(b) \in \mathcal{M},$
$$
\Delta_{\cS}(b)=f^{\prime}_{\cS}(b)=f^{\prime}_{\mathcal{M}}(g(b))g^{\prime}(b) =f^{\prime}_{\mathcal{M}}(a)g^{\prime}(b)=\Delta_{\cM}(a)g^{\prime}(b).
$$
\end{proof}

Let $\F^{\ast}_q=\langle \theta \rangle$. In many known results the set $\cS$ is chosen as a subgroup $ \gC=\langle \theta^e \rangle$
of $\F^{\ast}_q$ where $q-1=ef,$ or a coset $\cS=\theta^{\lambda} \gC$ of $\gC$ in $\F^{\ast}_q$
or an union of coset $\cS=\bigcup^t\limits_{i=1}(\theta^{\lambda_i} \gC) \ (0 \leq \lambda_1 < \lambda_2 <\cdots < \lambda_t  \leq e-1).$
For $\cS=\theta^{\lambda}C, f_{\cS}(x)=\prod^{f-1}\limits_{j=0}(x-\theta^{\lambda+ej})=x^f -\theta^{\lambda f}$
is a binomial and $f^{\prime}_{\cS}(x)=f x^{f-1}.$ Another candidate of $\cS$ is $\F_r$-subspace $\mathrm{V}$
of $\F_q,$ where $\F_r$ is a subfield of $\F_q \ (r=p^s, q=r^l=p^{sl}),$ a coset $\mathrm{V}+c \ (c \in \F_q)$ or an union of cosets.
It is well-known that
$$
f_{\mathrm{V}}(x)=\prod_{a \in \mathrm{V}}(x-a)=x^{r^t}+c_1x^{r^{t-1}}+\cdots+c_{t-1}x^r+c_tx \in \F_q[x]
$$
is a $\F_r$-linearized polynomial in $\F_q[x]$ and $ f^{\prime}_{\mathrm{V}}(x)=c_t$ where $t=\dim_{\F_r}\mathrm{V}.$

\begin{definition}
Let $\F_r$ be a subfield of $\F_q, r=p^s, q=r^l.$ A $\F_r$-linearized polynomial in $\F_q[x]$ has the following form

\begin{equation}
L(x)=c_0 x^{r^t}+c_1 x^{r^{t-1}}+ \cdots + c_{t-1}x^r+c_t x \in \F_q[x] \quad (c_0 \neq 0).
\end{equation}
\end{definition}

Let $L(x)$ be a $\F_r$-linearized polynomial in $\F_q[x]$. The mapping
$$
\varphi_{L}: \F_q \longrightarrow \F_q, \alpha \longmapsto L(\alpha)
$$
is $\F_r$-linear. Namely, for $a,b \in \F_r, \alpha, \beta \in \F_q,$
$$ L(a \alpha +b \beta)=a L(\alpha)+b L(\beta).$$
Thus the kernel and image of $\varphi_{L}$
$$ Ker(\varphi_{L})=\{c \in \F_q: L(c)=0\}, \ \ Im(\varphi_{L})=\{ L(c): c \in \F_q\} $$
are $\F_r$-subspaces of $\F_q, \dim_{\F_r} Ker(\varphi_{L}) +\dim_{\F_r} Im(\varphi_{L})=\dim_{\F_r}\F_q=l.$

For more facts on linearized polynomials we refer to the book \cite{LN}, Section 3.4.

\section{Review on Some Known Results}\label{sec-three}

In this section, we select few known results on MDS self-dual codes constructed via GRS codes and EGRS codes.
We present simple proofs with the approach illustrated in Section II. For more complete list of known
MDS self-dual codes we refer to the table in \cite{FLLL} and \cite{Yan}.

\begin{theorem}(\cite{FLLL})\label{thm-exam1}
Let $r=p^m \ (p \geq 3), q=r^2, 1 \leq l \leq m$ and $d=\gcd(l,m).$
Then for each $k, 1 \leq k \leq p^d, kp^l \in \Sigma(g,q)$ if $k$ is even
and $kp^l +1 \in \Sigma(eg,q)$ if $k$ is odd.
\end{theorem}

\begin{proof}
Let $m=dm^{\prime},l=dl^{\prime}, \F=\F_s$ for $s=p^d.$ Then $\F$ is a subfield of $\F_r$
and $\dim_{\F} \F_r = \frac{m}{d}=m^{\prime}.$ From $ 1 \leq l \leq m,$ we know that $  1 \leq l^{\prime} \leq m^{\prime}.$
We take a $\F$-subspace $\mathrm{V}$ of $\F_r$ with $\dim_{\F}\mathrm{V}=l^{\prime}.$ Then
$|\mathrm{V} |=|\F|^{l^{\prime}}=p^l \leq p^m=|\F_r|,$ and
$$
f_{\mathrm{V}}(x)=\prod_{c \in \mathrm{V}}(x-c)= x^{s^{l^{\prime}}}+c_1 x^{s^{l^{\prime}-1}}+ \cdots + c_{l^{\prime}-1}x^s+c_{l^{\prime}} x
$$
is a $\F$-linearized polynomial in $\F_r [x].$ Take $\gamma \in \F_q \backslash \F_r,$ then $f_{\mathrm{V}}(\gamma) \neq 0.$
For any $b \in \F, f_{\mathrm{V}}(b \gamma)=b f_{\mathrm{V}}(\gamma).$
Therefore $b \gamma +\mathrm{V} \ (b \in \F)$ are $p^d$ distinct cosets of $\mathrm{V}$ in $\F_q.$
We take a subset $\mathcal{M}=\{b_1, b_2, \ldots, b_k \}$ of $\F$ with size $k \ (\leq |\F|=p^d).$
Let $H_i=b_i \gamma+\mathrm{V}$ and $\mathcal{S}=\bigcup^k\limits_{i=1}H_i =\{ b_i  \gamma +c:  c \in \mathrm{V},1 \leq i \leq k \}.$
Then $|\mathcal{S}|=k |\mathrm{V}|=kp^l, $ and
\begin{eqnarray*}
f_{\cS}(x) &=& \prod^k_{i=1} \prod_{a \in H_i}(x-a)=\prod^k_{i=1} \prod_{c \in \mathrm{V}}(x-b_i \gamma -c)
=\prod^k_{i=1}f_{\mathrm{V}}(x-b_i \gamma) \\
&=& \prod^k_{i=1}(f_{\mathrm{V}}(x)-b_i f_{\mathrm{V}}(\gamma))=g (f_{\mathrm{V}}(x))
\end{eqnarray*}
where $g(x)=\prod^k\limits_{i=1}(x-b_i f_{\mathrm{V}}(\gamma)).$ Thus for each $a \in \mathcal{S}, a=b_{\lambda} \gamma+c, 1 \leq \lambda \leq k$
and $c \in \mathrm{V},$ we have $f_{\mathrm{V}}(a)=b_{\lambda}f_{\mathrm{V}}(\gamma)$ and by Lemma \ref{thm-three} (3),
$$
\Delta_{\cS}(a)=g^{\prime}(f_{\mathrm{V}}(a))f^{\prime}_{\mathrm{V}}(a)
=\Delta_{\mathcal{M^{\prime}}}(b_{\lambda}f_{\mathrm{V}}(\gamma))c_{l^{\prime}}
$$
where $\mathcal{M^{\prime}}=\{b_{\lambda} f_{\mathrm{V}}(\gamma): 1 \leq \lambda \leq k \}$ and then
$$
\Delta_{\mathcal{M^{\prime}}}(b_{\lambda}f_{\mathrm{V}}(\gamma))
=\prod^k \limits_{\scriptstyle i=1\atop \scriptstyle  i \neq \lambda}(b_{\lambda}f_{\mathrm{V}}(\gamma)-b_{i} f_{\mathrm{V}}(\gamma))
=f_{\mathrm{V}}(\gamma)^{k-1}\Delta_{\mathcal{M}}(b_{\lambda})
$$
where $\mathcal{M}=\{b_1, b_2,\ldots,b_k\}$. Since $b_i \in \F \subseteq \F_r \ (1 \leq i \leq k),$
we have $\Delta_{\mathcal{M}}(b_{\lambda})c_{l^{\prime}} \in \F^{\ast}_{r}$
and $\eta_q(\Delta_{\cS}(a))
=\eta_q(f_{\mathrm{V}}(\gamma)^{k-1}\Delta_{\mathcal{M}}(b_{\lambda})c_{l^{\prime}})
=\eta_q(f_{\mathrm{V}}(\gamma))^{k-1}.$
If $k$ is even, $\eta_q(\Delta_{\cS}(a))=\eta_q(f_{\mathrm{V}}(\gamma))$
is the same for all $a \in \cS.$
From Theorem \ref{thm-main} (1) we get $kp^l=|\cS| \in \Sigma(g,q).$
If $k$ is odd, then $k-1$ is even and then $\eta_q(- \Delta_{\cS}(a))=1$ for all $a \in \cS.$
From Theorem \ref{thm-main} (2), we get $kp^l+1 \in \Sigma(eg,q).$
\end{proof}

\begin{corollary}(\cite{Yan})
Let $q=p^{2m} \ (p \geq 3).$ Then for all even $n, n \equiv 0 \ \mbox{or} \ 1 \ (\bmod ~p^m)$
and $p^m +1 \leq n \leq q+1,$ we have $n \in \Sigma(q).$
\end{corollary}

\begin{proof}
Take $d=m$ in Theorem \ref{thm-exam1}.
\end{proof}

In Theorem \ref{thm-exam1}, $\cS$ is chosen as an union of cosets of a subspace of $\F_q.$
Now we consider $\cS$ being an union of cosets of a (cyclic) subgroup of $\F^{\ast}_q=\langle \theta \rangle.$
Let $q-1=ef,$ then $\gC=\langle \theta^{e} \rangle=\{ \theta^{ej}: 0 \leq j \leq f-1\}$
is a subgroup of $\F^{\ast}_q$ with $| \gC |=f.$ All cosets of $\gC$ in $\F^{\ast}_q$
are the $e$-th cyclotomic classes $D_{\lambda}=\theta^{\lambda}\gC \ (0 \leq \lambda \leq e-1).$

Let $\cS_i=\xi_i \gC \ (1 \leq i \leq t)$ be $t$ distinct cosets of
$\gC$ in $\F^{\ast}_q \ (0 \leq t \leq e-1), \cS=\bigcup^t\limits_{i=1}\cS_i, |\cS|=tf.$
Then
\begin{equation}\label{eqn-len}
f_{\cS}(x)=\prod_{a \in \cS}(x-a)=\prod^t_{i=1}\prod^{f-1}_{j=0}(x-\xi_i \theta^{ej})=\prod^t_{i=1}(x^f-\xi^f_{i})=g(x^f)
\end{equation}
where $g(x)=\prod^t\limits_{i=1}(x-\xi^f_{i})=f_{\cS^{\prime}}(x), \cS^{\prime}=\{\xi^f_{i}: 1 \leq i \leq t\}.$

In \cite{FLLL}, $\{\xi_{i}: 1 \leq i \leq t \}$ is taken as a subset of $\langle \theta^{r-1} \rangle$ for $q=r^2.$
Firstly we should determine that how many elements $\xi_{1},\cdots,\xi_{t}$ in $\langle \theta^{r-1} \rangle$
can be taken such that the cosets $\xi_{i}\gC \ (1 \leq i \leq t)$ are distinct.

\begin{lemma}\label{thm-con}
Let $q=r^2, r=p^m \ (p \geq 3), q-1=ef, \F^{\ast}_q=\langle \theta \rangle,
\gC=\langle \alpha \rangle, \mathcal{M}=\langle \beta \rangle,$
where $\alpha=\theta^e$ and $\beta=\theta^{r-1}.$
Let $\{\beta^{i_1},\beta^{i_2},\ldots,\beta^{i_t}\}$ be a subset of $\mathcal{M}$ where $i_1,i_2,\ldots,i_t$
are distinct module $r+1.$ Then $\beta^{i_{\lambda}}\gC \ (1 \leq \lambda \leq t)$ are distinct cosets
of $\gC$ in $\F^{\ast}_q$ if and only if $i_1,i_2,\ldots,i_t$ are distinct module
$\frac{r+1}{\gcd(r+1,f)}.$ Particularly, if $ 1 \leq t \leq \frac{r+1}{\gcd(r+1,f)},$ there exist $t$ elements
$i_{\lambda} \ (1 \leq \lambda \leq t)$ in  $\gZ_{r+1}=\mathbb{Z} / (r+1) \mathbb{Z}$ such that $\beta^{i_{\lambda}}\gC \ (1 \leq \lambda \leq t)$ are distinct.
\end{lemma}

\begin{proof}
For $i,j \in \gZ_{r+1},$
\begin{eqnarray*}
\beta^i \gC=\beta^j \gC & \Leftrightarrow & \beta^{i-j}=\theta^{(r-1)(i-j)} \in \gC=\langle \theta ^e \rangle \\
& \Leftrightarrow &  (r-1)(i-j) \equiv 0 \  (\bmod ~e) \Leftrightarrow \frac{e}{\gcd(e,r-1)} \mid  i-j.
\end{eqnarray*}
But $ef=q-1=r^2 -1$ and
$$ \frac{e}{\gcd(e,r-1)}=\frac{r^2 -1}{f} / \gcd(\frac{r^2 -1}{f}, r-1)=\frac{r^2 -1}{\gcd(r^2 -1, f(r-1))}=\frac{r+1}{\gcd(r+1,f)}.$$
Therefore $\beta^i \gC=\beta^j \gC$ if and only if $i  \equiv j \  (\bmod ~\frac{r+1}{\gcd(r+1,f)}).$
\end{proof}

Now suppose that $\{ i_1,i_2,\ldots,i_t\}$ be a subset of $\gZ_{r+1}$ such that $i_1,i_2,\ldots,i_t$
are distinct module $\frac{r+1}{\gcd(r+1,f)}.$
Let $ \mathrm{B}=\{ \beta^{i_{\lambda}} : 1 \leq \lambda \leq t\}, \beta=\theta^{r-1}.$ Then
\begin{equation}\label{eqn-sss}
\cS=\mathrm{B} \gC=\bigcup^t_{\lambda=1}\beta^{i_{\lambda}}\gC,
\end{equation}
is an union of $t$ cosets of $\gC$ in $\F^{\ast}_q$
and $ \mid \cS \mid=\mid \mathrm{B} \mid \mid \gC \mid=tf.$
For each $\gamma=\beta^{i_{\mu}}\alpha^j \in \cS, $ by (\ref{eqn-len}), we get
\begin{eqnarray}\label{eqn-mult}
\Delta_{\cS}(\gamma)&=& f^{\prime}_{\cS}(\gamma)=g^{\prime}(\gamma^f) f \gamma^{f-1} \nonumber \\
&=& f\beta^{i_{\mu}(f-1)} \theta^{je(f-1)}\Delta_{\cS^{\prime}}(\beta^{i_{\mu}f}) \quad \cS^{\prime}=\{\beta^{i_1 f},\cdots, \beta^{i_t f}\}  \nonumber \\
&=& f\beta^{i_{\mu}(f-1)} \theta^{-je}\Delta_{\cS^{\prime}}(\beta^{i_{\mu}f}).
\end{eqnarray}
From $\beta^r=\theta^{r(r-1)}=\theta^{1-r}=\beta^{-1},$ we get
\begin{eqnarray*}
\Delta_{\cS^{\prime}}(\beta^{i_{\mu}f})^r &=& \prod^t \limits_{\scriptstyle \lambda=1 \atop \scriptstyle \lambda \neq \mu}(\beta^{i_{\mu}fr}-\beta^{i_{\lambda}fr})=\prod^t \limits_{\scriptstyle \lambda=1 \atop \scriptstyle \lambda \neq \mu}(\beta^{-i_{\mu}f}-\beta^{-i_{\lambda}f})\\
&=& \prod^t \limits_{\scriptstyle \lambda=1 \atop \scriptstyle \lambda \neq \mu}
\frac{\beta^{i_{\lambda}f}-\beta^{i_{\mu}f}}{\beta^{i_{\mu}f}\beta^{i_{\lambda}f}}
=(-1)^{t-1}\beta^{-i_{\mu}f(t-2)-fI}\Delta_{\cS^{\prime}}(\beta^{i_{\mu}f})
\end{eqnarray*}
where $I=\sum^t\limits_{\lambda=1}i_{\lambda}.$
Therefore $\Delta_{\cS^{\prime}}(\beta^{i_{\mu}f})^{r-1}=\theta^A,$
where $$A=\frac{1}{2}(t-1)(r^2-1)-f(r-1)[i_{\mu}(t-2)+I]$$
and then
\begin{equation}\label{eqn-shor}
\Delta_{\cS^{\prime}}(\beta^{i_{\mu}f})=\theta^{B+(r+1)s}
\end{equation}
where $s \in \mathbb{Z}$ and
\begin{equation}\label{eqn-thr}
B=A/(r-1)=\frac{1}{2}(t-1)(r+1)-f[i_{\mu}(t-2)+I].
\end{equation}
Since $2 \nmid r, f \in \F^{\ast}_r, \eta_q(\beta)=\eta_q(\theta)^{r-1}=1$ and $ \eta_q(\theta)=-1,$
from (\ref{eqn-mult}), \ (\ref{eqn-shor}) and (\ref{eqn-thr}), we get $\eta_q(\Delta_{\cS}(\gamma))=(-1)^d,$
\begin{equation}\label{eqn-ddaa}
d=ej+\frac{1}{2}(t-1)(r+1)+f(i_{\mu}t+I).
\end{equation}
Moreover, let $\widetilde{\cS}=\cS \bigcup \{0 \}.$ Then $| \widetilde{\cS}|=tf+1$
and for $\gamma=\beta^{i_{\mu}}\alpha^j \in \cS, \Delta_{\widetilde{\cS}}(\gamma)=\Delta_{\cS}(\gamma)\gamma.$
By (\ref{eqn-mult}), \ (\ref{eqn-ddaa}) and $\eta_q(\gamma)=(-1)^{ej}$ we get
\begin{equation}\label{eqn-ddd}
\eta_q (\Delta_{\widetilde{\cS}}(\gamma))=(-1)^{\tilde{d}}, \quad \tilde{d}=\frac{1}{2}(t-1)(r+1)+f(i_{\mu}t+I).
\end{equation}
For $ 0 \in \widetilde{\cS},$
$$
\Delta_{\widetilde{\cS}}(0)=(-1)^{tf}\prod^t_{\lambda=1}\prod^{f-1}_{j=0}(\beta^{i_{\lambda}}\alpha^j)
=(-1)^{tf}\beta^{fI}\alpha^{\frac{f(f-1)}{2}}
$$
and
\begin{equation}\label{eqn-etad}
\eta_q (\Delta_{\widetilde{\cS}}(0))=(-1)^{ef(f-1)/2}=1 \quad \mbox{since} \ ef=q-1=r^2 -1 \equiv 0 \  (\bmod ~4).
\end{equation}

With above preparation we introduce the following result which appear in \cite{FLLL} but
with slight different statement.

\begin{theorem}\label{thm-crit3}
Let $r$ be a power of a prime number $p \geq 3, q=r^2, q-1=ef$ and $R=\frac{ r+1 }{\gcd(r+1,f)}.$

(I) Assume that $tf$ is even and $ 1 \leq t \leq R$

$(I_1)$ If $e$ is even, then $tf \in \Sigma(g,q);$

$(I_2)$ $tf+2 \in \Sigma(eg,q)$ if

(A) $2 \mid f$ and $4 \mid (t-1)(r+1);$ or

(B) $2 \nmid f.$

(II) Assume that $tf$ is odd. Then for $ 1 \leq t \leq R/2, tf+1 \in \Sigma(eg,q).$
\end{theorem}

\begin{proof}
Let $\cS$ be the set defined by (\ref{eqn-sss}), $\widetilde{\cS}=\cS \bigcup \{0 \}.$

(I) If $2 \mid tf,$ by (\ref{eqn-ddaa}) we know that for
$\gamma=\beta^{i_{\mu}}\alpha^j \in \cS, \eta_q(\Delta_{\cS}(\gamma))=(-1)^d,$ where
$$
d=ej+f t i_{\mu}+f I+\frac{1}{2}(t-1)(r+1)  \equiv e j+fI +\frac{1}{2}(t-1)(r+1) \  (\bmod ~2).
$$

$(I_1)$ If $2 \mid e,$ then
$d \equiv fI +\frac{1}{2}(t-1)(r+1) \  (\bmod ~2)$ which is independent of $i_{\mu}$ and $j.$
Namely, all $\eta_q(\Delta_{\cS}(\gamma)) \ (\gamma \in \cS)$ are the same. By Theorem \ref{thm-main} (1)
we get $tf \in \Sigma(g,q).$

$(I_2)$ From (\ref{eqn-etad}) we get $\eta_q(-\Delta_{\widetilde{\cS}}(0))=1.$ From (\ref{eqn-ddd}) we know that for
$\gamma=\beta^{i_{\mu}}\alpha^j \in \cS, \eta_q(-\Delta_{\widetilde{\cS}}(\gamma))=(-1)^{\tilde{d}} \ (\gamma \in \cS)$ where
$\tilde{d} \equiv fI +\frac{1}{2}(t-1)(r+1) \  (\bmod ~2 ).$
If $2 \mid f$ and $4 \mid (t-1)(r+1),$ then $2 \mid \tilde{d}$ and $\eta_q(-\Delta_{\widetilde{\cS}}(\gamma))=1$
for all $\gamma \in \cS.$ By Theorem \ref{thm-main} (2) we get $tf+2=| \widetilde{\cS}|+1 \in \Sigma(eg,q.)$
If $2 \nmid f, \tilde{d} \equiv I +\frac{1}{2}(t-1)(r+1) \  (\bmod ~2 ).$ When $1 \leq t \leq R-1,$ it is easy
to see that there exist $0 \leq i_1 < i_2 < \cdots i_t \leq R-1$ such that
 $$I=\sum^t\limits_{\lambda=1}i_{\lambda} \equiv \frac{1}{2}(t-1)(r+1) \  (\bmod 2 ).$$
 Then $\eta_q(-\Delta_{\cS}(\gamma))=1$ for all $\gamma \in \cS.$

 By Theorem \ref{thm-main} (2), we get $tf+2 \in \Sigma(eg,q).$ When $t=R,$
 then $\{i_1, i_2 , \cdots, i_t \}=\{0,1, \cdots, R-1\} (\bmod ~R ), I \equiv \frac{1}{2}R(R-1) \  (\bmod ~2 ),$
and $ \tilde{d} \equiv \frac{1}{2}R(R-1)+(R-1)(r+1) \  (\bmod ~2 ).$ Remark that $R=\frac{r+1}{\gcd(r+1,f)}$
is even for $2 \nmid f.$ Then we get
$$ 2 \tilde{d} \equiv R(R-1)+(R-1)(r+1) \equiv R+r+1 \equiv 0\  (\bmod ~4 ).$$
Therefore $\eta_q(-\Delta_{\widetilde{\cS}}(\gamma))=1$ for all $\gamma \in \cS.$
 By Theorem \ref{thm-main} (2), we get $tf+2 \in \Sigma(eg,q).$

(II) If $2 \mid tf,$ by (\ref{eqn-ddaa}) we get $\eta_q(-\Delta_{\cS}(\gamma))=(-1)^d$
for $\gamma=\beta^{i_{\mu}}\alpha^j \in \cS$ where $ d \equiv i_{\mu}+I \equiv 0 \  (\bmod ~2)$
and $I=\sum^t\limits_{\lambda=1}i_{\lambda}.$ From $2 \nmid f,$ we know that $2 \mid R.$
If $1 \leq t \leq R/2,$ we take $i_{\mu}=2 \mu -1 \ (1 \leq \mu \leq t).$
Then
$$ I=\sum^t\limits_{\mu=1}i_{\mu} \equiv t \equiv 1 \  (\bmod ~2), \quad
 d \equiv i_{\mu}+I \equiv 0 \  (\bmod ~2)$$
which means that $\eta_q(-\Delta_{\cS}(\gamma))=1$ for all $\gamma \in \cS.$
By Theorem \ref{thm-main} (2), we get $tf+1=1+|\cS| \in \Sigma(eg,q).$
\end{proof}

\begin{theorem}(\cite{FLLL})\label{thm-FLL}
Let $r$ be a power of a prime number $p \geq 3, q=r^2, q-1=ef, 2 \mid s \mid f, 2s \mid r+1$
and $D=\frac{s(r-1)}{\gcd (s(r-1),f)}.$ Then for any $t, 1 \leq t \leq D, tf+2 \in \Sigma(eg,q).$
Moreover, if $2 \mid e,$ then $tf \in \Sigma(g,q).$
\end{theorem}

\begin{proof}
Let $\F^{\ast}_q=\langle \theta \rangle,
\gC=\langle \alpha \rangle, \alpha=\theta^e, \beta=\theta^{\frac{r+1}{s}}.$
Similarly as Lemma \ref{thm-con}, we can prove that for
$ 1 \leq t \leq D, \cS=\bigcup^t\limits_{\lambda=1}\beta^{\lambda}\gC$
is a disjoint union of $t$ cosets of $\gC$ in $\F^{\ast}_q, |\cS|=tf.$
For $\gamma=\beta^{\mu}\alpha^j \in \cS,$ by (\ref{eqn-mult}) we get
\begin{equation}\label{eqn-del}
\Delta_{\cS}(\gamma)=f \beta^{\mu (f-1)} \theta^{-ej} \Delta_{\cS^{\prime}}(\beta^{\mu f})
\end{equation}
where $\cS^{\prime}=\{\beta^{\lambda  f}: 1 \leq \lambda \leq t\}.$ From $s \mid f$
we know that
$$(\beta^{\lambda  f})^r=\theta^{\frac{f}{s}\lambda r(r+1)}
=\theta^{\frac{f}{s}\lambda (r+1)}=\beta^{\lambda f}.$$
Therefore $\beta^{\lambda  f} \in \F^{\ast}_r,\Delta_{\cS^{\prime}}(\beta^{\mu f}) \in \F^{\ast}_r$
and $\eta_q(\Delta_{\cS^{\prime}}(\beta^{\mu f}))=1.$ From $2 \mid \frac{r+1}{s},$ we get $\eta_q(\beta)=1.$
Then by (\ref{eqn-del}),
\begin{equation}\label{eqn-eta}
\eta_q (\Delta_{\cS}(\gamma))=\eta_q(\theta^{ej})=(-1)^{ej}.
\end{equation}
If $e$ is even, then $tf=|\cS| \in \Sigma(g,q)$ by Theorem \ref{thm-main} (1).

On the other hand, for $\widetilde{\cS}=\cS \bigcup \{0 \},$ we have
\begin{eqnarray*}
\Delta_{\widetilde{\cS}}(0)&=&(-1)^{tf}\prod^t_{\lambda=1}\prod^{f-1}_{j=1}(\beta^{\lambda}\alpha^j) \\
\eta_q (\Delta_{\widetilde{\cS}}(0))&=& \eta_q(\alpha^{\frac{f(f-1)}{2}})=\eta_q(\theta^{\frac{(q-1)(f-1)}{2}})=1.
\end{eqnarray*}
Since $ q=r^2 \equiv 1 \  (\bmod ~4).$
For $\gamma=\beta^{\mu}\alpha^j \in \cS, \Delta_{\widetilde{\cS}}(\gamma)=\Delta_{\cS}(\gamma)\gamma$
and $\eta_q(\Delta_{\widetilde{\cS}}(\gamma))=\eta_q(\theta^{ej} \theta^{ej})=1$ for all $\gamma \in \cS.$
By Theorem \ref{thm-main} (2) we get $tf+2 \in \Sigma(eg,q).$
\end{proof}

\section{New Results}\label{sec-four}

In this section we present several new constructions of MDS self-dual codes. Firstly we consider
$\cS$ as a disjoint union of a subset and a subspace of $\F_q.$

\begin{theorem}\label{thm-new1}
Let $ r=p^m  \ (p \geq 3), q=r^2.$ Then for each even number $n, 2 \leq n \leq 2r,$
we have $n \in \Sigma(q).$ Precisely, let $ 0 \leq l \leq r-1,  d \in \{ 0, 1\}.$
Then $l+dr \in \Sigma(g,q)$ if $ 2 \mid l+d,$ and $l+dr+1 \in \Sigma(eg,q)$ if $2 \nmid l+d.$
\end{theorem}

\begin{proof}
For $d=0, 2 \leq l \leq r-1,$ this is Corollary \ref{thm-coro}. For $d=1,$ consider the trace
$\tr(x)=x+x^r$ from $\F_q$ to $\F_r.$ The set $V$ is comprised of the zeros of $\tr(x)$ is a $\F_r$-subspace
of $\F_q, |V|=r.$ For each $a \in \F_r, \tr(a)=2a$ which implies that $\F_r \bigcap V=\{ 0\}.$
Let $\cM =\{ a_1, a_2,\ldots, a_l\}$ be a subset of
$\F^{\ast}_r, 0 \leq l \leq r-1.$
 Let $\cS=\mathcal{M} \bigcup V,$
then $\mathcal{M} \bigcap V=\emptyset, |\cS|=l+r$ and
$$f_{\cS}(x)=\prod_{\alpha \in \cS}(x-\alpha)=f_{\mathcal{M}}(x)f_{V}(x)$$
where $f_{\mathcal{M}}(x)=\prod\limits_{ a \in \mathcal{M}}(x-a) \in \F_r[x]$ and $f_{V}(x)=\tr(x).$
Thus $f^{\prime}_{\cS}(x)=f^{\prime}_{\mathcal{M}}(x)\tr(x)+f_{\mathcal{M}}(x).$

For $\gamma \in \cM, \Delta_{\cS}(\gamma)=f^{\prime}_{\cS}(\gamma)=\Delta_{\cM}(\gamma)\tr(\gamma) \in \F^{\ast}_r.$
For $\gamma \in V, \Delta_{\cS}(\gamma)=f^{\prime}_{\cS}(\gamma)=f_{\cM}(\gamma) \in \F^{\ast}_r.$ Thus
$\eta_q(f^{\prime}_{\cS}(\gamma))=1$ for all $\gamma \in \cS.$ The conclusion is derived from Theorem \ref{thm-main}.
\end{proof}

For all MDS self-dual codes over $\F_q$ constructed so far, $q$ is a square. Now we show a
``lifting'' result without this restriction.

\begin{theorem}\label{thm-new2}
Let $q=p^m, Q=q^l, l \geq 1$ and $p \geq 3.$

(1) If $n$ is even, $ 2 \leq n \leq q-1$ and $ n \in \Sigma(g,q),$ then $ nq^{l-1} \in \Sigma(g,Q).$

(2) If $n$ is odd, $1 \leq n \leq q$ and $n+1 \in \Sigma(eg,q),$ then $ nq^{l-1}+1 \in \Sigma(eg,Q).$
\end{theorem}

\begin{proof}
(1) Suppose that $2 \mid n, 2 \leq n \leq q-1$ and $ n \in \Sigma(g,q).$ Then there exists a
subset $\cM$ of $\F_q, |\cM|=n$ such that $\eta_q(\Delta_{\cM} (a))$ are the same for all $a \in \cM$.
Consider the trace mapping
$$\tr: \F_Q \longrightarrow \F_q,  \quad \tr(\alpha)=\alpha+\alpha^q+\alpha^{q^2}+\cdots+\alpha^{q^{l-1}}.$$
This is a surjective and $\F_q$-linear mapping. Let $\theta \in \F_Q$ such that $\tr(\theta)=1.$
Let $V$ be the set of the zeros of $\tr(x).$ Then $V$ is a $\F_q$-subspace of $\F_Q, |V|=q^{l-1}.$
For each $a \in \cM,$ the set $\cS_a=\{b \in \F_Q: \tr(b)=a\}$ is the coset $a\theta +V$ of $V$ in $\F_Q.$
Let $\cS=\bigcup\limits_{a \in \cM}\cS_a.$ Then $|\cS|=|\cM||\cS_a|=nq^{l-1}$ and
\begin{eqnarray*}
f_{\cS}(x) &=& \prod_{a \in \cM}\prod_{b \in \cS_{a}}(x-b)=\prod_{a \in \cM}\prod_{c \in V}(x-a\theta -c)\\
&=&\prod_{a \in \cM}\tr(x-a\theta)=\prod_{a \in \cM}(\tr(x)-a)=g(\tr(x))
\end{eqnarray*}
where $g(x)=\prod\limits_{a \in \cM}(x-a)=f_{\cM}(x).$ Therefore for $b=a\theta +c \in \cS_a, a \in \cM, c \in V,$
$$\Delta_{\cS}(b)=f^{\prime}_{\cS}(b)=f^{\prime}_{\cM}(\tr(b))=f^{\prime}_{\cM}(a)=\Delta_{\cM}(a).$$
Since $\eta_Q(\Delta_{\cS}(b))=\eta_Q(\Delta_{\cM}(a))=\eta_q(\Delta_{\cM}(a))^{\frac{Q-1}{q-1}} \ (b \in \cS)$
are the same, we get $n q^{l-1}=|\cS| \in \Sigma(g,Q)$ by Theorem \ref{thm-main} (1).

(2) Suppose that $2 \nmid n, 1 \leq n \leq q$ and $n+1 \in \Sigma(eg,q).$ Then there exists a subset $\cM$
of $\F_q, |\cM|=n$ such that $\eta_q(-\Delta_{\cM}(a))=1$ for all $a \in \cM.$ As in (1), let $\tr$
be the trace mapping from $\F_Q$ to $\F_q, \theta, V$ and $\cS_a \ (a \in \cM)$ are the same as in (1).
For $\cS=\bigcup\limits_{a \in \cM}\cS_a,$ we also have $\Delta_{\cS}(b)=\Delta_{\cM}(a)$ for $b \in \cS.$
Since
$$
\eta_Q(-\Delta_{\cS}(b))=\eta_Q(-\Delta_{\cM}(a))=\eta_q(-\Delta_{\cM}(a))^{\frac{Q-1}{q-1}}=1
$$
for all $a \in \cM,$ we get $nq^{l-1} +1 \in \Sigma(eg,q)$ by Theorem \ref{thm-main} (2).
\end{proof}

In the next construction we use the norm mapping in stead of the trace.
Let $\F_r$ be the subfield of $\F_q, q=r^s.$ The norm mapping for extension $\F_q/\F_r$ is
$$\textrm N  (x):\F^{\ast}_q \longrightarrow \F^{\ast}_r, \quad \textrm N(\alpha)=\alpha^{\frac{q-1}{r-1}}.$$

This is a surjective homomorphism of (multiplicative) groups. Let $\F^{\ast}_q=\langle \theta \rangle,$
the set of zeros of $ \textrm N (x)= x^{\frac{q-1}{r-1}}=1$ is the subgroup $B=\langle \theta^{r-1}\rangle$ of
$\F^{\ast}_q$ and for each $a \in \F^{\ast}_r,$
$$ \textrm N^{-1}(a)=\{ b \in \F^{\ast}_q: \textrm N(b)=b^{\frac{q-1}{r-1}}=a \}$$
is a coset $b^{\prime}B$ where $b^{\prime}$ is any element in $\F^{\ast}_q$
such that $(b^{\prime})^{\frac{q-1}{r-1}}=a.$

\begin{theorem}\label{thm-new3}
Let $r=p^m \ (p \geq 3)$ and $q=r^s.$

(1) If $s$ is even, and $ 1 \leq l \leq \frac{r-1}{2}$, then $l \cdot \frac{q-1}{r-1} \in \Sigma(g,q)$
and $l \cdot \frac{q-1}{r-1}+2 \in \Sigma(eg,q).$

(2) Suppose that $s$ is odd, $ 1 \leq l \leq r-1$

(2.1) If $l$ is even and $l \in \Sigma(g,r),$ then $l \cdot \frac{q-1}{r-1} \in \Sigma(g,q);$

(2.2) If $l$ is odd and $l+1 \in \Sigma(eg,r),$ then $l \cdot \frac{q-1}{r-1}+1 \in \Sigma(eg,q).$
\end{theorem}

\begin{proof}
Let $\rm N$ $(x)=x^{\frac{q-1}{r-1}}, \F^{\ast}_q=\langle \theta \rangle, B=\langle \theta^{r-1} \rangle,
\cM =\{ a_1, a_2, \cdots, a_l\} \subseteq \F^{\ast}_r,  1 \leq  l  \leq (r-1)/2.$
Let $b_i \in \F^{\ast}_q$ such that $\rm N$$(b_i)=a_i \ (1 \leq i \leq l).$ Let $\cS_i=b_i B$
and $\cS=\bigcup^{l}\limits_{i=1}\cS_i.$
Then
$$|\cS|=l \cdot \frac{q-1}{r-1} \equiv ls \  (\bmod ~2 ), \quad   \
f_{\cS}(x)=f_{\cM}(x^{\frac{q-1}{r-1}})$$
where $f_{\cM}(x)=\prod^l\limits_{i=1}(x-a_i).$
For each $\alpha \in \cS_i, \alpha=b_i c \ (c \in B),$
$$ \Delta_{\cS}(\alpha)=\frac{q-1}{r-1} f^{\prime}_{\cM}(\alpha^{\frac{q-1}{r-1}}) \alpha^{\frac{q-1}{r-1}-1}
=f^{\prime}_{\cM}(\alpha^{\frac{q-1}{r-1}})\alpha^{\frac{q-1}{r-1}-1}.$$
Since $\alpha^{\frac{q-1}{r-1}}=b^{\frac{q-1}{r-1}}_i=a_i, \frac{q-1}{r-1}-1 \in \F^{\ast}_p \subseteq \F^{\ast}_q$ and $\frac{q-1}{r-1}-1 \equiv s-1 \  (\bmod ~2),$ we have
$$ \Delta_{\cS}(\alpha)= \Delta_{\cM}(a_i)\alpha^{\frac{q-1}{r-1}-1}, \quad
\eta_q(\Delta_{\cS}(\alpha))=\eta_q(\Delta_{\cM}(a_i))\eta_q(\alpha)^{s-1}.$$
Then from $\alpha=b_i c, c \in \langle \theta^{r-1} \rangle, 2 \mid r-1$ and
$\Delta_{\cM}(a_i) \in \F^{\ast}_r$ we get
\begin{equation}\label{eqn-etaq}
\eta_q(\Delta_{\cS}(\alpha))=\eta_q(\Delta_{\cM}(a_i))\eta_q(b_i)^{s-1}
=\eta_r(\Delta_{\cM}(a_i))^s\eta_q(b_i)^{s-1}.
\end{equation}

(1) Suppose that $2 \mid s.$ Then $\eta_q(\Delta_{\cS}(\alpha))=\eta_q(b_i)$
for $\alpha=b_i c \in \cS_i.$ Let $\mathcal{Q}=\{a \in \F^{\ast}_r : \eta_r(a)=1\}.$
Then $|\mathcal{Q}|=\frac{r-1}{2},$ and by assumption $1 \leq l \leq \frac{r-1}{2},$
we can take a subset $\cM=\{a_1,a_2,\cdots, a_l\}$ of $\mathcal{Q}, a_i =e^2_i (e_i \in \F^{\ast}_r).$
Let $b^{\prime}_{i} \in \F^{\ast}_q$ such that $\textrm N (b^{\prime}_{i})=e_i.$
Then for $b_i=(b^{\prime}_{i})^2, \textrm N$ $(b_i)=e^2_i=a_i$ and
$\eta_q(b_i)=\eta_q(b^{\prime}_i)^2 =1\ (1 \leq i \leq l).$
Therefore $\eta_q(\Delta_{\cS}(\alpha))=1$ for all $\alpha \in \cS,$
and $|\cS| = ls \equiv 0 \  (\bmod ~2 ).$ By Theorem \ref{thm-main} (1) we get
$l \cdot \frac{q-1}{r-1} \in \Sigma(g,q).$

Moreover, let $\widetilde{\cS}=\cS \cup \{ 0 \}, | \cS|=l \cdot \frac{q-1}{l-1}+1$ is odd. By $2 \mid s$
we get $\eta_q(-\Delta_{\widetilde{\cS}}(\alpha))=1.$ For $\alpha=b_i c \ (c \in B=\langle \theta^{r-1} \rangle)$ and $ 2 \mid r-1,$ we get $\eta_q(\alpha)=1.$ Therefore
$\eta_q(-\Delta_{\widetilde{\cS}}(0))=\eta_q(\prod\limits_{\alpha \in \cS} \alpha)=1.$
 By Theorem \ref{thm-main} (2), we get $l \cdot \frac{q-1}{l-1}+2 \in \Sigma(eg,q).$

(2) Suppose that $2 \nmid s.$ Then $l \cdot \frac{q-1}{r-1} \equiv ls \equiv l \  (\bmod ~2 )$
and by (\ref{eqn-etaq}) we have $\eta_q(\Delta_{\cS}(\alpha))=\eta_r(\Delta_{\cM}(a_i))$
for $\alpha =b_i c \ (c \in B=\langle \theta^{r-1} \rangle ),$ $\textrm N(b_i)=a_i.$

(2.1) If $2 \mid l$ and $l \in \Sigma(g,r),$ there exists a subset
$\cM=\{a_1,a_2,\cdots, a_l\}$ of $\F_r$ such that $\eta_r(\Delta_{\cM}(a)) \ (a \in \cM)$
are the same. For any $a \in \F_r,$ let $a+\cM=\{a+a_1, a+a_2, \cdots, a+a_l\}.$
It is easy to see that $\Delta_{\cM}(a_i)=\Delta_{a+\cM}(a+a_i).$ Then by assumption
$1 \leq l \leq r-1,$ we can choose $\cM$ being a subset of $\F^{\ast}_r.$ Then by Theorem \ref{thm-main} (1)
and $\eta_q (\Delta_{\cS}(\alpha))=\eta_r(\Delta_{\cM}(a_i)) \ (\alpha=b_i c,$ $\textrm N$ $(b_i)=a_i)$
we get $$l \cdot \frac{q-1}{r-1}=|\cS| \in \Sigma(g,q).$$

(2.2) If $2 \nmid l$ and $l+1 \in \Sigma(eg,r),$ there exists a subset
$\cM=\{a_1,a_2,\cdots, a_l\}$ of $\F^{\ast}_r$ such that $\eta_r(-\Delta_{\cM}(a))=1$
for all $a \in \cM.$ Then $\eta_q(-\Delta_{\cS}(\alpha))=\eta_r(-\Delta_{\cM}(a))=1$
for all $\alpha \in \cS \ (\alpha=bc, c \in B, b^{\frac{q-1}{r-1}}=a \in \cM).$
By Theorem \ref{thm-main} (2) we get $l \cdot \frac{q-1}{r-1}+1 \in \Sigma(eg,q).$
\end{proof}

\section{Conclusion and Open Problems}\label{sec-five}

Basic on fundamental results ( Theorem \ref{thm-main} (1) and (2) ) on MDS self-dual codes constructed via generalized
RS codes and extended generalized RS codes given by Jin and Xing \cite{JX} and Yan \cite{Yan} respectively, we
 present an unified approach to treat previous constructions with simplified proofs and
concise statements and show several new constructions of MDS self-dual codes. Now we raise two open problems.

(1) To determine the set $\Sigma(g,q),\Sigma(eg,q)$ and $\Sigma(q)$ for certain $q=p^m \ (p \geq 3 \ \mbox{and} \ m \geq 1).$
Many even numbers in these three sets are known by previous constructions of MDS self-dual codes over $\F_q.$
On the other hand, the following result shows that there exist at least a half of even number which do not belong to
$\Sigma(q)$ for $q \equiv 3  (\bmod ~4).$

\begin{theorem}\label{thm-five}
If $q=p^m \equiv 3  (\bmod ~4)$ and $n \equiv 2  (\bmod ~4),$ then there is no self-dual codes over $\F_q$
with length $n.$ Particulary, $n \notin \Sigma(q).$
\end{theorem}

\begin{proof}
Suppose that $\cC$ is a self-dual code over $\F_q$ with length
$n=2k, k=\dim_{\F_q}\cC, \cC \subseteq \F^n_q.$ Up to an equivalent, we assume that the first $k$
symbols are information symbols. Then $\cC$ has a generator matrix in the following form
$$
\begin{gathered}
G=
\begin{bmatrix}
I_k & P
\end{bmatrix}=
\begin{bmatrix}
v_1 \\
\vdots \\
v_k
\end{bmatrix},
\quad
P=(p_{ij})_{1 \leq i,j \leq k}=\begin{bmatrix}
u_1 \\
\vdots \\
u_k
\end{bmatrix}
\end{gathered}
$$
$$v_i=(e_i | u_i),e_i=(0,\cdots,0,1,0,\cdots,0), u_i=(p_{i1},p_{i2},\cdots,p_{ik}) \ (1 \leq i \leq k)$$
the $i$-th places of $e_i$ is $1$, and other place of $e_i$ are $0.$ Since $\cC$ is self-dual, we get, for
$ 1 \leq i,j \leq k,0=v_i v^{T}_j=e_i e^T_j+u_i u^T_j=\delta_{ij}+u_i u^T_j, $
$$
\delta_{ij} =  \left \{
\begin{array}{ll}
1, & \mbox{if} \ i=j, \\
0, & \mbox{otherwise}.
\end{array}
\right.
$$
This means that $PP^T=-I_k,$ so that $(\det P)^2=(-1)^k.$ By assumption $q \equiv 3  (\bmod ~4)$
and $k=\frac{n}{2}\equiv 1  (\bmod ~2),$ we get $1=\eta_q (\det P)^2=\eta_q (-1)^k=-1,$
a contradiction. Therefore there is no self-dual code over $\F_q$ with length $n \equiv 2  (\bmod ~4).$
\end{proof}

The following conjecture is well-known.

\textbf{MDS Main Conjecture:} For all MDS codes over  $\F_q \ (2 \nmid q),$ the length $n$ is at most $q+1.$
This conjecture has been proved in \cite{Ball} when $q=p$ is an odd prime number.

(2) Is the MDS Main Conjecture true for MDS self-dual codes over $\F_q$ where $q=p^m,$ and $ m \geq 2?$

%\section*{Acknowledgements}

% K.Feng's research was supported by the Natural Science Foundation of China under Grant No:11571107 and 11471178.
% and Tsinghua National Lab. for Information Science and Technology.
% A.Zhang's research was supported by the Natural Science Foundation of China under Grant No:11401468.

\end{document}